\documentclass{article}
\usepackage[margin=1.25in]{geometry}
\usepackage{amsmath, amssymb, amsthm, mathtools,amscd,mathalpha,amsfonts,mathrsfs,tikz-cd}
\usepackage{algorithm2e}
\usepackage{listings}
\usepackage{enumerate}
\usepackage[parfill]{parskip}
\usepackage{xcolor}
\usepackage{appendix}
\usepackage{graphicx}
\usepackage{caption}
\usepackage{subcaption}
\usepackage[font={it}]{caption}
\usepackage{float}
\usepackage{caption}
\usepackage{amsmath}
\usepackage[hidelinks]{hyperref}
\usepackage{mathabx}
\usepackage[
style=alphabetic,
]{biblatex}

\usepackage[]{mdframed}

\newtheorem{theorem}{Theorem}
\newtheorem{corollary}{Corollary}
\newtheorem*{corollary*}{Corollary}
\newtheorem*{theorem*}{Theorem}
\newtheorem{prop}{Proposition}
\newtheorem{lemma}{Lemma}

\newtheorem*{lemma*}{Lemma}

\newtheorem{definition}{Definition}

\newcommand{\todo}[1]{{\color{red}[TODO: #1]}}
\newcommand{\eps}{\varepsilon}

\DeclareMathOperator\spn{span}
\DeclareMathOperator\poly{poly}

\DeclareMathOperator\R{\mathbb{R}}
\DeclareMathOperator\C{\mathbb{C}}
\DeclareMathOperator\tr{tr}
\DeclareMathOperator\supp{supp}

\DeclarePairedDelimiter{\angles}{\langle}{\rangle}
\usepackage{physics}

\DeclareMathOperator\bS{\boldsymbol{S}}
\DeclareMathOperator\bJ{\boldsymbol{J}}
\DeclareMathOperator\bD{\boldsymbol{\Delta}}
\DeclareMathOperator\tbD{\tilde{\bD}}
\DeclareMathOperator\bH{\boldsymbol{H}}
\DeclareMathOperator\tbH{\tilde{\bH}}

\DeclareMathOperator\mP{\mathcal{P}}
\DeclareMathOperator\mG{\mathfrak{G}}
\DeclareMathOperator\md{\mathfrak{d}}
\DeclareMathOperator\mD{\mathcal{D}}
\DeclareMathOperator\mC{\mathcal{C}}

\DeclareMathOperator\BH{\mathcal{B}(\mathcal{H})}

\usepackage{authblk}
\usepackage[normalem]{ulem}

\title{
Certified algorithms for quantum Hamiltonian learning \\
via energy-entropy inequalities
}
\author[1]{Adam Artymowicz}
\author[2]{Hamza Fawzi}
\author[3]{Omar Fawzi}
\author[2]{Samuel O. Scalet}
\affil[1]{California Institute of Technology, Pasadena, USA}
\affil[2]{Department of Applied Mathematics and Theoretical Physics, University of Cambridge, United Kingdom}
\affil[3]{{Univ Lyon, Inria, ENS Lyon, UCBL, LIP, Lyon, France}}
\date{}

\begin{document}

\maketitle

\begin{abstract}
We consider the problem of learning the Hamiltonian of a quantum system from estimates of Gibbs-state expectation values. Various methods for achieving this task were proposed recently, both from a practical and theoretical point of view. On the one hand, some practical algorithms have been implemented and used to analyze experimental data but these algorithms often lack correctness guarantees or fail to scale to large systems. On the other hand, theoretical algorithms with provable asymptotic efficiency guarantees have been proposed, but they seem challenging to implement. Recently, a semidefinite family of Hamiltonian learning algorithms was proposed which was numerically demonstrated to scale well into the 100-qubit regime, but provided no provable accuracy guarantees. We build on this work in two ways, by extending it to provide certified \textit{a posteriori} lower and upper bounds on the parameters to be learned, and by proving \textit{a priori} convergence in the special case where the Hamiltonian is commuting.
\end{abstract}

\section{Introduction}
In this work we consider the problem of Hamiltonian learning from a thermal state. Given the form of the Hamiltonian
\begin{equation}
    h=\sum_{\alpha=1}^m \lambda_\alpha E_\alpha
\end{equation}
where local Hamiltonian terms $E_\alpha$ are known but their coefficients $\lambda_\alpha$ are not,
we seek estimates of the $\lambda_\alpha$'s using measurements from the Gibbs state
\[
\rho=\frac{e^{-\beta h}}{\Tr[e^{-\beta h}]}.
\]
Hamiltonian learning is fundamental to validate our models of quantum physical systems~\cite{wiebe2014hamiltonian,wang2017experimental}, and plays an important role in certifying quantum devices. In addition, with the recent progress in quantum algorithms for Gibbs state preparation~\cite{temme2011quantum,chen2023quantum}, Hamiltonian learning is likely to be crucial to benchmark the future realizations of such algorithms. Aside from these applications, it can also directly be used to give physical insights about many-body quantum systems in both experimental and numerical contexts. For instance, Hamiltonian learning has already been applied to the study of \textit{entanglement Hamiltonians} \cite{Dalmonte2022}, which are sensitive probes of entanglement and have for instance been used to test CFT predictions in thermalizing systems~\cite{Wen2018, kokail2021entanglement}. Applications like these demand a Hamiltonian learning algorithm that is efficient, reliable, and can provide rigorous bounds on the uncertainty.



\paragraph{Overview of the algorithm and results}
We propose an efficient Hamiltonian learning algorithm that provides rigorous error bounds.
Our algorithm is based on a semidefinite constraint that generalizes a set of 
inequalities known as the \textit{energy-entropy balance} or \textit{EEB} inequalities \cite{arakisewell}. These constraints were first applied to the forward problem in \cite{fawzi2023} by three of us and to the inverse problem in \cite{adam-preprint} by one of us. In the latter work, it was benchmarked numerically, where it was found to scale well and to give accurate results when reconstructing a known Hamiltonian. However, this work did not include a theoretical analysis of the case with measurement noise, or any practically useful convergence guarantees. As such its output did not come with any guarantees of accuracy, which would be necessary for use in experiments. Here we solve this problem in two ways, with \textit{a posteriori} and \textit{a priori} guarantees.

Our algorithm computes rigorous lower and upper bounds on any linear functional $v\cdot \lambda$ of the unknown coefficients $\lambda \in \R^m$. By varying over different choices of $v$ a convex relaxation of the Hamiltonian parameters can be obtained, i.e., a convex set including $\lambda$ itself.
In particular, by running the algorithm for all basis vectors in $\R^m$, we obtain intervals such that $\lambda_\alpha\in[a_\alpha,b_\alpha]$.
Other interesting choices exist, however: if $\sum_\alpha v_\alpha E_\alpha$ is a symmetry of $\rho$, the quantity $v\cdot \lambda$ has the interpretation of a generalized \textit{chemical potential} for this symmetry \cite{Araki-chemical-potential}, an important physical quantity. This example includes the usual chemical potential when $\sum_\alpha v_\alpha E_\alpha$ is the number operator corresponding to a given particle type.

The algorithm is parametrized by a hierarchy level $\ell \in \mathbb{N}_+$ which governs the strength of the semidefinite constraint used. It is summarized in Figure~\ref{tbl:algorithmIntro}. The details of the setup of the algorithm are described in Section~\ref{sec:framework-and-main-results} and the theorem is proved in Sections~\ref{sec:noise-stability} and~\ref{sec:convergence}.

\begin{figure}[h]\label{tbl:algorithmIntro}
\begin{mdframed}
\begin{center}
\textbf{Algorithm A (overview)}
\end{center}
\begin{itemize}
    \item {\bf Input}:
    \begin{itemize}
        \item $v \in \R^m$ -- coefficients of the linear functional $v\cdot \lambda$
        \item $\tilde{\omega}$ -- (noisy) oracle to Gibbs state observables 
        \item $\eps_0$ -- upper bound on the estimate of the noise
        \item $\ell \in \mathbb{N}_+$ -- level of the hierarchy
    \end{itemize}
    \item {\bf Output}: Pair of numbers $a^{(\ell)}$, $b^{(\ell)}$\\
    \item[1.] Evaluate $\tilde{\omega}$ on all operators of the form $PQ$ and $P [E_{\alpha}, Q]$ for all $P,Q \in \mP_{k,\ell}$ ($\mP_{k,\ell}$ is a set of local Pauli operators defined in Definition~\ref{def:pkl}) and $\alpha=1,\ldots,m$ and collect the results in the matrices $\tilde{\bD}$ and $\tilde{\bH}_{\alpha}$ according to Equations \eqref{eqn:bH-alpha-definition}. Compute $\mu_1$ and $\mu_2$ according to Theorem~\ref{thm:a-posteriori-feasibility}.
    \item[2.] Solve the semidefinite programs:
\begin{align}
& \underset{\substack{\lambda'\in \R^m}}{\text{minimize/maximize}}
& & v\cdot \lambda'  \\
& \text{subject to}
& & \log(\tilde{\bD}) + \sum_{\alpha=1}^m\lambda'_\alpha(\tilde{\bH}_{\alpha} + \tilde{\bH}_{\alpha}^\dagger)/2  \succeq -\mu_1, \\ 
& && \pm i\sum_{\alpha=1}^m\lambda'_\alpha(\tilde{\bH}_{\alpha} - \tilde{\bH}_{\alpha}^\dagger)/2 \preceq \mu_2. 
\end{align}
    \item[3.] Return $a^{(\ell)}$ and $b^{(\ell)}$, the minimum and maximum values of the program respectively.
\end{itemize}
\end{mdframed}
\caption{Algorithm to estimate parameters of local Hamiltonians from noisy observables of Gibbs states. We note that our presentation is slightly different from other works in Hamiltonian learning: instead of giving access to the copies of the state $\rho$, we have access to a noisy oracle computing expectation values of $\rho$. It is clear that having access to $O(\frac{1}{\eps_0^2})$ copies of $\rho$, we can obtain an estimate of expectation values within $\eps_0$ with high probability. We chose this presentation to avoid having probabilistic statements.
\label{tbl:algorithmIntro}
}
\end{figure}




\begin{theorem}[Informal version of Theorem~\ref{thm:algoA-guarantees}]
Consider a $k$-local Hamiltonian $h=\sum_{\alpha=1}^m \lambda_{\alpha} E_\alpha$ with unknown coefficients $\lambda \in \R^m$. Let $\beta = \max_{\alpha} |\lambda_{\alpha}|$ and let $\md$ be the degree of the dual interaction graph (see Section \ref{sec:framework-and-main-results} for the definition). Assume we have access to an oracle $\tilde{\omega}$ such that for teh evaluated observables $O$, $|\tilde{\omega}(O) - \tr[\rho O]| \leq \eps_0 \|O\|$, where $\rho = e^{-h}/\Tr[e^{-h}]$ is the Gibbs state of $h$.
Let $v \in \R^{m}$ with $\|v\|_1=1$ and consider the problem of estimating the inner product $v\cdot \lambda$ given access to $\tilde{\omega}$.

For each $\ell \in \mathbb{N}_+$, let $a^{(\ell)}$ and $b^{(\ell)}$ be the output of algorithm in Figure~\ref{tbl:algorithmIntro}.


    \begin{enumerate}[i)]
        \item The pair of numbers $a^{(\ell)}$ and $b^{(\ell)}$ satisfy
        \begin{align}
            a^{(\ell)} \le v\cdot \lambda \le b^{(\ell)}. \label{eqn:a-posteriori-bound-B-in-intro}
        \end{align}
        \item If $h$ is commuting and $\ell = \max(3,1+(1+\md)^2)$, then we have
        \begin{align}
            b^{(\ell)}-a^{(\ell)} \le \eps_0/\sigma
        \end{align} 
        provided $\eps_0 \leq \sigma$ for some error threshold $\sigma = e^{-\mathcal{O}_{k,\md}(\beta)}m^{-6}$. 
    \end{enumerate}
\end{theorem}

Point (i) above shows that our algorithm returns \textit{certified a posteriori bounds}, i.e., it returns an estimate of the parameter $v\cdot \lambda$ together with strict bounds on the error of these estimates, namely $b^{(\ell)} - a^{(\ell)}$.
This type of algorithm is particularly beneficial for situations where no convergence guarantees can be proven or when \textit{a priori} guarantees lead to overly pessimistic bounds. Moreover, if the program is infeasible, the dual program gives a certificate that $\rho$ is not a Gibbs state of any Hamiltonian with the given structure. 
Point (ii) gives an a priori guarantee: for a constant value of $\ell$ and provided the noise $\eps_0$ of the expectation values is below a certain threshold, the estimates of $v \cdot \lambda$ returned by the algorithm are proportional to $\eps_0$. In particular, estimates can be obtained in \textit{polynomial time} and using \textit{polynomially many samples} in the system size and inverse error of the estimates.

\paragraph{Related work}
The version of the Hamiltonian learning problem we study has attracted wide interest over the past years, with several approaches being proposed.
On the practical front, reliable algorithms which have been used in numerics and experiments have been limited to small system sizes \cite{Anshu2021, kokail2021entanglement,lifshitz2021practical}. More scalable approaches exist, including algorithms that only use the property that the Gibbs state is a steady state~\cite{bairey2019learning,evans2019scalable}. Even though such algorithms are efficient and can work for some instances, one can easily construct instances where such conditions are not enough to single out the Hamiltonian, even when the Hamiltonian is commuting and the expectation values are known perfectly \cite{adam-preprint}. We mention also the Hamiltonian learning  algorithm~\cite{gu2022practical} (see also~\cite{stilck2024efficient}) which has both a theoretical analysis of performance and has been implemented. However, it solves a different problem where Gibbs states with different temperatures can be queried.



From a theoretical perspective, the polynomial complexity of the commuting case follows from the simple approach in \cite{anshu-commuting}. Other polynomial complexity bounds have been proven in this and other contexts \cite{Haah_2024, kuwahara2024clusteringconditionalmutualinformation,bakshi2023,narayanan_improved_2024}. The best asymptotic computational and sample complexity bounds for the general problem are achieved in \cite{bakshi2023,narayanan_improved_2024}, which give an algorithm with polynomial classical and sample complexity. It is interesting to compare their algorithm with ours. They use a convex hierarchy that is based on relaxations of the KMS condition \cite{Haag1967}, which is physically related to the EEB condition used in this work. Both the KMS and the EEB conditions are interpretable in terms of local thermodynamic stability \cite{Cohen1987, arakisewell}, and their ideal versions (i.e., taking all the possible conditions) have been shown to be equivalent \cite{arakisewell,Sewell1977}. However, in order to obtain an efficient algorithm, only a subset of the conditions are imposed and this manifests differently for the KMS versus EEB conditions.
Since the KMS conditions are non-linear in the Hamiltonian parameters, in \cite{bakshi2023}, polynomial approximations are used to implement these constraints entailing an intricate error analysis already in the feasibility part.
Furthermore, to deal with the resulting polynomial constraint systems, an additional hierarchy, the sum-of-squares relaxation, is introduced.
In comparison, our constraints, which are linear in the parameters of the Hamiltonian, are much easier to implement.

\section{Algorithmic framework and main results}\label{sec:framework-and-main-results}
Let $\mathcal{H}$ be a finite-dimensional Hilbert space. We write $\BH$ for the algebra of all linear operators on $\mathcal{H}$. The adjoint of an operator $a$ will be denoted $a^*$. Let $m>0$ and let $E_1,\hdots, E_m \in \BH$ be a collection of selfadjoint operators with $\|E_\alpha\|=1$ for all $1\le \alpha \le m$. Let
\begin{align}
    h = \sum_{\alpha=1}^m\lambda_\alpha E_\alpha
\end{align}
for some unknown coefficients $\lambda_1,\hdots, \lambda_m$, and set $\beta := \max_{\alpha=1,\hdots m}|\lambda_\alpha|$. Let $\rho$ be the thermal state
\[
\rho =\frac{e^{-h}}{\Tr[e^{-h}]}.
\]
We will write thermal expectation values of observables as
\[
\omega(A) :=\Tr[\rho A].
\]
Here and below, we use $\Tr$ to denote the usual trace and $\tr$ to denote the normalized trace $\tr(a)=\Tr(a)/\dim(\mathcal{H})$.

\subsection{The matrix EEB constraint}

We start by introducing the semidefinite constraint that forms the backbone of this work.
It begins with a choice of selfadjoint operators $P_1,\hdots, P_r \in \BH$ satisfying $\|P_i\|=1$ which we call the \textit{perturbing operators}. The constraint will depend on this choice, and adding operators to the list will yield a tighter constraint -- in this sense we will say the constraints form a hierarchy  depending on the choice of $P_1,\hdots , P_r$\footnote{The fact that adding perturbing operators strengthens the constraint was shown in a slightly different setting in \cite[Proposition 4]{adam-preprint}. There, it was also shown that the constraint depends only on the \textit{span} of the perturbing operators, and thus we have a hierarchy indexed by the poset of linear subspaces of $\BH$. We do not need either of these facts here.}. Later, in sections \ref{sec:a-priori} and \ref{sec:algo} we will restrict to lattice models, where the perturbing operators will be chosen to be all the Pauli operators of a given locality, but in this section and the next we keep the choice of perturbing operators unrestricted.

Define the following $r\times r$ matrix $C$ and $r\times r$ matrices $B_1,\hdots, B_m$:
\begin{align}
    C_{ij} &:= \omega(P_iP_j) &&1\le i,j\le r\\
    (B_\alpha)_{ij} &:= \omega(P_i[E_\alpha,P_j]) && 1\le i,j\le r \text{ and } 1\le \alpha \le m.
\end{align}
Since $(P_iP_j)^*=P_jP_i$, the matrix $C$ is automatically hermitian, and in fact it is guaranteed to be positive-definite because $\rho\succ0$. This allows us to define the following matrices:
\begin{align}
    \bD &:= C^{-1/2}C^TC^{-1/2},\\
    \bH_{\alpha} &:= C^{-1/2}B_{\alpha}C^{-1/2} \hspace{10mm} 1\le \alpha \le m. \label{eqn:bH-alpha-definition}
\end{align}
We remark that although the \textit{operators} $E_\alpha$ are hermitian, the \textit{matrices} $B_\alpha$ and $\bH_\alpha$ are not hermitian in general.

The conceptual starting point for this work is the following matrix inequality \cite{fawzi2023}:
\begin{theorem}[Matrix EEB inequality]\label{thm:mEEB-ideal}
\begin{align}
    \log(\bD) + \sum_{\alpha=1}^m\lambda_\alpha\bH_\alpha &\succeq 0. \label{eqn:mEEB-ideal}
\end{align}
\end{theorem}
Notice that since positive semidefinite matrices are by definition hermitian, the matrix EEB inequality subsumes the linear constraint $\sum_{\alpha=1}^m\lambda_\alpha\left(\bH_\alpha-\bH_\alpha^\dagger\right) = 0$, which is nontrivial because the matrices $\bH_\alpha$ are (in general) not hermitian. 

Varying the choice of perturbing operators $P_1,\hdots ,P_r$, the matrix EEB inequality yields a hierarchy of semidefinite contraints that are satisfied for the true Hamiltonian coefficients $\lambda = (\lambda_1,\hdots, \lambda_m)$. This is an idealized constraint using the quantities $\bD$ and $\bH_\alpha$ built out of noise-free expectation values, which are not accessible to experiment. This is more than a practical concern, since shot noise is unavoidable given access to only finitely many copies of the state $\rho$, even in the absence of other sources of error like state preparation and measurement noise. Thus, in order to make practically useful statements we will need to consider an appropriate relaxation of this constraint, which we do below.

\subsection{Relaxing the constraint}
We assume access to an estimate of $\omega(A)$, which we call $\tilde{\omega}(A)$, for every observable $A$ of the form 
\begin{align}
    &P_iP_j && 1\le i,j\le r, \label{eqn:operator1}\\
    &P_i[E_\alpha, P_j]  && 1\le i,j\le r \text{ and } 1\le \alpha \le m. \label{eqn:operator2}
\end{align}
We assume the following control over the errors in the estimates $\tilde{\omega}$: for some $\epsilon_0\ge0$ we have
\begin{align}
    |\tilde{\omega}(A) - \omega(A)|\le \epsilon_0 \hspace{5mm} \text{for any operator $A$ of the form (\ref{eqn:operator1}) or (\ref{eqn:operator2}).}
\end{align}
We further assume that the estimates obey the following restrictions, which may be enforced without loss of generality:
\begin{align}
    \tilde{\omega}(1)&=1, \label{eqn:omegatilde-assum1} \\
    |\tilde{\omega}(P_iP_j)| &\le 1 \hspace{5mm} \text{for all $1\le i,j \le r$} \label{eqn:omegatilde-assum2}\\
    |\tilde{\omega}(P_i[E_\alpha,P_j])| &\le 2 \hspace{5mm} \text{for all $1\le i,j \le r$ and $1\le \alpha \le m$} \label{eqn:omegatilde-assum3}.
\end{align}

Define the quantities $\tilde{C}, \tilde{B}_{\alpha}, \tbD$, $\tbH_\alpha$ analogously to $C,B_{\alpha},\bD, \bH_\alpha$, but using the noisy estimates $\tilde{\omega}$ in place of $\omega$. In order to do this, one needs $\tilde{C}$ to be positive-definite. We assume this for now, but as we will show in Theorem \ref{thm:a-priori-feasibility} below, this is guaranteed for sufficiently low noise levels. We relax the matrix EEB inequality as follows:
\begin{align}
 \log(\tbD) + \sum_{\alpha=1}^m\lambda'_\alpha(\tbH_{\alpha} + \tbH_{\alpha}^\dagger)/2  &\succeq -\mu_1, \label{linear-constraint-noisy} \\ 
\pm i\sum_{\alpha=1}^m\lambda'_\alpha(\tbH_{\alpha} - \tbH_{\alpha}^\dagger)/2 &\preceq \mu_2. \label{SDP-constraint-noisy}
\end{align}
for $\mu_1,\mu_2\ge 0$. To show that this relaxation is useful, we need to find some relaxation parameters $\mu_1,\mu_2$ for which the true Hamiltonian coefficients are feasible. The following theorem is proved in Section \ref{sec:noise-stability}:
\begin{theorem}[A posteriori feasibility]\label{thm:a-posteriori-feasibility}
    Suppose $\tilde{C}\succ 0$ and let $K := 2r\|\tilde{C}^{-1}\|$. If $\eps_0\le 1/K$ then the true Hamiltonian coefficients $\lambda'=\lambda$ are a feasible point for the relaxed matrix EEB constraints (\ref{linear-constraint-noisy}) and (\ref{SDP-constraint-noisy}) with
    \begin{align}
        \mu_1 &:=  \left(2K^3 + 3m\beta K^2 \right)\eps_0 \label{eqn:mu1}\\
        \mu_2 &:= 3m\beta K^2\eps_0. \label{eqn:mu2}
    \end{align}
\end{theorem}
We call this an \textit{a posteriori} feasibility guarantee because it depends on the estimates $\tilde{\omega}$, which enter through the constant $K$.
Thanks to this result, the inequalities (\ref{linear-constraint-noisy}) and (\ref{SDP-constraint-noisy}) with  $\mu_1$ and $\mu_2$ given by (\ref{eqn:mu1}) and (\ref{eqn:mu2}) place rigorous constraints on the set of potential Hamiltonian parameters $\lambda'$.
If, for a given ansatz, the condition $\eps_0\le1/K$ is fulfilled, but the constraints are infeasible, this \textit{guarantees} that the given state is not the Gibbs state in the family of Hamiltonians.

However, this guarantee only holds when $\eps_0\le 1/K$ and it is not \textit{a priori} clear that this condition can be satisfied by ensuring $\eps_0$ is small enough because $K$ depends on the measured data. The following Lemma shows that this is indeed the case if we assume that the perturbing operators are chosen to be Hilbert-Schmidt orthogonal:
\begin{lemma}\label{lem:a-priori-feasibility-general}
    Suppose the $P_1,\hdots, P_r$ satisfy $\tr(P_iP_j)=\delta_{ij}$ and let $\sigma = e^{-m\beta}d/3r$, where $d=\dim(\mathcal{H})$. Then if $\eps_0\le \sigma$ then $K\le 1/\sigma$ and in particular, Theorem \ref{thm:a-posteriori-feasibility} holds.
\end{lemma}
\begin{proof}
Since $\rho \succeq e^{-\|h\|} \succeq e^{-m\beta}$ we have $w^\dagger C w
    \ge  e^{-m\beta}d\|w\|^2$
for any $w\in \C^r$, and so $C\succeq e^{-m\beta}d = 3r\sigma$. An elementary bound gives $\|C-\tilde{C}\|\le r\eps_0$, so if $\eps_0\le \sigma$ then $\|C-\tilde{C}\|\le r\sigma$. It follows that $\tilde{C}\ge 2r\sigma$, and so $K = 2 r \|\tilde{C}\| \le 1/\sigma$.
\end{proof}
The exponential dependence on $m$ is unavoidable in the general case. However, as we show in the next section, if we restrict to the setting of local Hamiltonians on a lattice, we can prove the analogous statement with an error thrshold $\sigma$ that does not depend on $m$ or the Hilbert space dimension $d$ at all.

\subsection{The lattice setting}\label{sec:a-priori}

So far, we have made almost no assumptions on the structure of the problem and were able to give the \textit{a posteriori} guarantee in Theorem \ref{thm:a-posteriori-feasibility} and the \textit{a priori} guarantee in Lemma \ref{lem:a-priori-feasibility-general}. In this section, we restrict to the setting of a quantum lattice system, where locality allows us to give much stronger \textit{a priori} guarantees.

Suppose our Hilbert space is that of a collection of $n$ qubits $\mathcal{H}= (\C^2)^{\otimes n}$. We say an operator $E$ is $k$-\textit{supported} if $|\supp(E)|\le k$, and $k$-\textit{local} if it is a sum of $k$-supported terms.
Suppose the Hamiltonian terms $\{E_\alpha\}_{a=1}^m$ are $k$-supported for some $k>0$, and define their \textit{dual interaction graph} $\mG$ to have the vertex set $[m]$ and an edge between every pair of vertices $1\le \alpha, \alpha'\le m$ with $\supp(E_\alpha)\cap \supp(E_b)\neq\emptyset$. We say $\{E_1,\hdots , E_m\}$ are \textit{$k$-$\md$-low-intersection} if the degree of the graph $\mG$ is bounded by $\md$.
Call an operator $F$ $k$-$\ell$-$\mG$-supported if $\supp(F)\subset \bigcup_S \supp(E_\alpha)$ for a set $S\subset \mG$ which satisfies $|S|\le\ell$ and is connected in $\mG$. Note that a $k$-$\ell$-$\mG$-supported operator is $k\ell$-supported.
\begin{definition}
\label{def:pkl}
     We define $\mP_{k,\ell}$ to be the set of all $k$-$\ell$-$\mG$-supported Paulis.
\end{definition}
In what follows, we will treat $k$ and $\md$ as constants that does not depend on the system size $m$. This scenario includes familiar examples of geometrically local Hamiltonians defined on a lattice.

The following theorem gives an \textit{a priori} feasibility guarantee, by proving that Theorem~\ref{thm:a-posteriori-feasibility} holds for sufficiently low error rate $\eps_0$ without the need to incorporate measurement outcomes (which previously entered through the constant $K$) into the error threshold. The proof appears in Section \ref{sec:noise-stability}.
\begin{theorem}[A priori bound on $K$]\label{thm:a-priori-feasibility}
    Suppose that $E_\alpha$ are $k$-$\md$-low-intersection for some constants $k,\mathfrak{d}$, and suppose that the operators $P_1,\hdots, P_r$ are all $k$-$\ell$-$\mG$-supported for some $\ell>0$. Then there is an error threshold
        \begin{align}
            \sigma = r^{-2}e^{-\mathcal{O}_{k,\md,\ell}(\beta)}
        \end{align}
        such that $\epsilon_0\le \sigma$ implies $K\le 1/\sigma$. In particular, Theorem \ref{thm:a-posteriori-feasibility} holds.
\end{theorem}
Next, we give an \textit{a priori} convergence result, i.e. a proof that our algorithm outputs estimates of the Hamiltonian parameters that are close to the true Hamiltonian parameters, in the case when the true Hamiltonian is \textit{commuting}. The class of commuting Hamiltonians contains many interesting examples, including topologically ordered Hamiltonians.
\begin{definition}\label{def:commuting}
    We say a Hamiltonian $h = \sum_{\alpha} \lambda_\alpha E_\alpha$ is \emph{commuting} if there are selfadjoint operators $\{F_\alpha\}_{\alpha=1}^{m}$ with $\|F_\alpha\|=1$ and real constants $\{\nu_\alpha\}_{\alpha=1}^{m}$ such that 
    \begin{enumerate}[i)]
        \item $h = \sum_{\alpha=1}^{m}\nu_\alpha F_\alpha$,
        \item $[F_\alpha,F_{\alpha'}] = 0$ for any $1\le \alpha,\alpha'\le m$,
        \item $\supp{F_\alpha} \subset \supp{E_\alpha}$ for all $1\le \alpha \le m$.
        \item $\max_{\alpha = 1,\hdots m}|\nu_\alpha| \le \mathcal{C}\max_{\alpha = 1,\hdots m}|\lambda_\alpha|$ for some constant $\mathcal{C}>0$. 
    \end{enumerate}
\end{definition}
Note that for the purpose of our algorithm it is not necessary to \textit{know} the decomposition in Definition~\ref{def:commuting}.
A general ansatz of Pauli operators is used for the $\{E_\alpha\}_{\alpha=1}^m$, and Theorem \ref{thm:commuting-convergence} holds as long as a commuting decomposition exists.
In order to prove a convergence guarantee, it is also necessary to ensure that the set $P_1,\hdots P_r$ is large enough, which we have not done so far.
We require them to include all Pauli operators satisfying a $k$-$\ell$-$\mG$-support assumption for a constant $\ell$.
The following Theorem is proven in Section~\ref{sec:commuting-convergence}

\begin{theorem}[A priori convergence in the commuting case]\label{thm:commuting-convergence}
    Suppose $E_\alpha$ are Pauli operators with $k$-$\mathfrak{d}$-low-intersection and $h$ is commuting. Set $\{P_1,\hdots, P_r\} = \mP_{k,\ell}$ for $\ell = \max(3,1+(\mathfrak{d}+1)^2)$. There is an error threshold
    \begin{align}
        \tau &= m^{-6}e^{-\mathcal{O}_{k,\md,\mC}(\beta)}
    \end{align}
    such that if $\epsilon_0 \le \tau$ then for any $\lambda'\in \R^m$ satisfying the constraints (\ref{linear-constraint-noisy}) and (\ref{SDP-constraint-noisy}), we have
    \begin{align}
        \sup_{\alpha=1,\hdots m}|\lambda_\alpha'-\lambda_\alpha| \le e^{\mathcal{O}_{k,\md,\mC}(\beta)}\left(\mu_1+m^{1/2}\mu_2\right) + \eps_0/\tau.
    \end{align}
\end{theorem}

\subsection{Algorithms}\label{sec:algo}
Let us now describe how the above constraints lead to two algorithms for Hamiltonian learning.
The first algorithm takes as input the set of perturbing operators $P_1,\ldots,P_r$, a choice of direction in the parameter space $v\in \R^m$, choices of nonnegative numbers $\mu_1,\mu_2$, and the estimates $\tilde{\omega}$. Assumptions on the errors in the estimates will enter through the choice of $\mu_1$ and $\mu_2$.
\vspace{0.5cm}
\begin{mdframed}
\begin{center}
\textbf{Algorithm A}
\end{center}
Compute $K:= 2r\|\tilde{C}^{-1}\|$. If $\tilde{C}$ is not positive definite or $K > 1/\eps_0$, return the interval $[-\infty,+\infty]$. Otherwise, let
\begin{align}
        \mu_1 &:= \left(2K^3 + 3m\beta K^2 \right)\eps_0\\
        \mu_2 &:= 3m\beta K^2\eps_0,
\end{align}
and return the interval $[a,b]$ where $a,b$ are minimum/maximum of the following SDP:
\begin{align}
& \underset{\substack{\lambda'\in \R^m}}{\text{minimize/maximize}}
& & v \cdot \lambda' \\
& \text{subject to}
& & \log(\tbD) + \sum_{\alpha=1}^m\lambda'_\alpha(\tbH_{\alpha} + \tbH_{\alpha}^\dagger)/2  \succeq -\mu_1,  \\ 
& && \pm i\sum_{\alpha=1}^m\lambda'_\alpha(\tbH_{\alpha} - \tbH_{\alpha}^\dagger)/2 \preceq \mu_2. \label{algo2-sdp}
\end{align}
\end{mdframed}
Theorems \ref{thm:a-posteriori-feasibility}, \ref{thm:a-priori-feasibility}, and \ref{thm:commuting-convergence} give the following guarantees:

\begin{theorem}[Guarantees for Algorithm A]\label{thm:algoA-guarantees}
    Suppose the input state $\rho$ is the Gibbs state of $h=\sum_\alpha \lambda_\alpha E_\alpha$ for some unknown coefficients $\lambda \in \R^m$ with $\max_{\alpha=1,\hdots, m}|\lambda_\alpha|\le \beta$, and Pauli operators $E_\alpha$ and consider Algorithm A for some $v\in \R^m$ with $\|v\|_1=1$. Then: 
    \begin{enumerate}[i)]
        \item The algorithm is feasible and its output $[a,b]$ satisfies
        \begin{align}
            a\le v\cdot \lambda \le b. \label{eqn:a-posteriori-bound-B}
        \end{align}
        \end{enumerate}
        If furthermore $P_1,\ldots,P_r=\mP_{k,\ell}$, then:
        \begin{enumerate}[i)]\setcounter{enumi}{1}
        \item If $h$ is commuting and $\ell = \max(3,1+(1+\md)^2)$ then there is an error threshold 
        \begin{align}
            \sigma = e^{-\mathcal{O}_{k,\md,\mC}(\beta)}m^{-6}
        \end{align} such that $\eps_0\le \sigma$ implies $b-a \le \eps_0/\sigma$.
    \end{enumerate}
\end{theorem}
Let us remark that in the above, increasing the error bound $\eps_0$ can never violate the assumptions of the algorithm, but it does relax the constraints, producing worse bounds. As such, the practicality of the algorithm depends on access to error bounds that are not too conservative, and in some settings, such tight error bounds may be hard to obtain. In these settings, one of course cannot expect tight bounds on the $\lambda_\alpha$, but one may be more interested in the qualitative structure of the Hamiltonian rather than in exact error bounds. 

For such settings, we introduce a variant of the above algorithm, which does not assume any prior knowledge of the measurement error. It requires only a choice of hierarchy level $\ell\in \mathbb{N}_+$ and the estimates $\tilde{\omega}$.\clearpage 
\begin{mdframed}
\begin{center}
\textbf{Algorithm B}
\end{center}
Assume $\tilde{C}\succ 0$. Return the optimal parameters $\mu\in \R_{\ge 0}$, $\lambda' \in \R^m$ of the following program:
\begin{align}
& \underset{\substack{\lambda'\in \R^m \\ \mu\in \R}}{\text{minimize}}
& & \mu \\
& \text{subject to}
& & \log(\tbD) + \sum_{\alpha=1}^m\lambda'_\alpha(\tbH_{\alpha} + \tbH_{\alpha}^\dagger)/2  \succeq -\mu, \label{algo1-linear} \\ 
& && \pm i\sum_{\alpha=1}^m\lambda'_\alpha(\tbH_{\alpha} - \tbH_{\alpha}^\dagger)/2 \preceq \mu. 
\end{align}
\end{mdframed}
The algorithm returns the putative Hamiltonian parameters $\lambda' \in \R^m$ and the parameter $\mu\ge 0$, which can be interpreted as a confidence parameter. A zero or near-zero value of $\mu$ indicates confidence that $\lambda'$ are the true Hamiltonian parameters (a \textit{positive result}) while a large value of $\mu$ indicates confidence that $\rho$ is not the Gibbs state of any Hamiltonian in the span of $E_1,\hdots, E_m$ (a \textit{negative result}). Using Theorems \ref{thm:a-posteriori-feasibility}, \ref{thm:a-priori-feasibility}, and \ref{thm:commuting-convergence}, we equip this algorithm with the following guarantees:

\begin{theorem}[Guarantees for Algorithm B]\label{thm:algoB-guarantees}
    Suppose the input state $\rho$ is the Gibbs state of $h=\sum_\alpha \lambda_\alpha E_\alpha$ for some unknown coefficients $\lambda \in \R^m$ with $\max_{\alpha=1,\hdots, m}|\lambda_\alpha|\le \beta$, and let $\mu\in \R_{\ge 0}$ and $\lambda'\in \R^m$ be the output of Algorithm B. As before, let $K:=2r\|\tilde{C}^{-1}\|$.
    \begin{enumerate}[i)]
        \item The algorithm is feasible and if $K\le 1/\eps_0$ returns $\mu \le \left(2K^3 + 3m\beta K^2 \right)\eps_0$.
        
        \item If furthermore $P_1,\ldots,P_r=\mP_{k,\ell}$, with $\ell = \max(3,1+(1+\md)^2)$ and $h$ is commuting then there is an error threshold 
        \begin{align}
            \sigma = e^{-\mathcal{O}_{k,\md,\mC}(\beta)}m^{-6}
        \end{align} such that $\eps_0\le \sigma$ implies $\sup_{\alpha=1,\hdots, m}|\lambda_\alpha'-\lambda_\alpha| \le \eps_0/\sigma$.
    \end{enumerate}
\end{theorem}
Part $i)$ is an a posteriori guarantee against \textit{false negatives}: if $K\le 1/\eps_0$ then the program will never terminate with a value $\mu > \left(2K^3 + 3m\beta K^2 \right)\eps_0$ if the true Hamiltonian is in the span of $E_1,\hdots, E_m$\footnote{Furthermore, a negative result is interpretable in terms of thermodynamic stability: the dual program produces Lindbladian that increases the entropy of $\rho$ while keeping the expectation values of $E_\alpha$ fixed for all $\alpha = 1,\hdots m$. This was shown in the noise-free case in \cite{adam-preprint}}. Part $ii)$ and $iii)$ show that the algorithm converges in the case of small systems (ie. $m$ constant) and commuting Hamiltonians, and that the convergence rate matches theoretical complexity bounds up to $\poly(m)$ factors.


%
%
%

\section{Feasibility proofs}\label{sec:noise-stability}
Now that we are finished stating the main results, we move on to the proofs. In this section we prove the main feasibility statements: Theorem \ref{thm:a-posteriori-feasibility} (a posteriori feasibility) and Theorem \ref{thm:a-priori-feasibility} (a priori feasibility). We will first prove a continuity bound for the matrix EEB inequality in terms of the condition number of the estimated correlation matrix $\tilde{C}$, and Theorem \ref{thm:a-posteriori-feasibility} will follow from the continuity bound and the matrix EEB inequality. Then Theorem \ref{thm:a-priori-feasibility} will follow from an \textit{a priori} bound on the condition number of $\tilde{C}$.
\paragraph{Continuity of the matrix EEB constraint}
We begin with some elementary lemmas.
A standard equivalence between finite-dimensional matrix norms goes as follows:
\begin{lemma}\label{lem:op-norm-bound}For an $n\times n$-matrix $A$
    \begin{equation}
    \|A\|\le n\max_{i,j}|a_{i,j}|
    \end{equation}
\end{lemma}
We will make use of the following continuity bounds for matrix functions.
\begin{lemma}\label{lem:matrix-continuity-bounds}
    Let $A,B,\widetilde A,\widetilde B$ be matrices. Then
    \begin{enumerate}[i)]
        \item If $A$ and $\tilde{A}$ are positive definite then
        \begin{align}
             \left\|\sqrt{A}-\sqrt{\tilde{A}}\right\|\le\frac{\|A-\tilde{A}\|}{\|A^{-1/2}\|^{-1}+\|\tilde{A}^{-1/2}\|^{-1}} \label{eqn:sqrt-continuity}
        \end{align}
        and
        \begin{align}\label{eqn:log-continuity}
        \|\log(A)-\log(B)\|\le\max\{\|A^{-1}\|,\|B^{-1}\|\}\|A-B\|
        \end{align}
        \item if $A$ and $B$ are invertible then
        \begin{align}
            \|A^{-1}-B^{-1}\|\le \|A^{-1}\|\|B^{-1}\|\|A-B\|.
        \end{align}
        \item 
        \begin{align}
            \|AB-\widetilde A\widetilde B\|\le\|A\|\|B-\widetilde B\|+\|\widetilde B\|\|A-\widetilde A\|
        \end{align}
    \end{enumerate}
\end{lemma}
\begin{proof}
    Inequality (\ref{eqn:sqrt-continuity}) is proven in   \cite[Lemma 2.2]{schmitt1992}. Inequality \eqref{eqn:log-continuity} follows from \cite[Theorem X.3.8]{bhatia1997}.
    The remaining inequalities are elementary from submultiplicativity and triangle inequality.
\end{proof}
Using the above lemmas, we will prove the main continuity bound:
\begin{prop}[Continuity bound for matrix EEB inequality]\label{prop:constraint-continuity-2}
    Let $K := 2r\|\tilde C^{-1}\|$ and suppose $\eps_0\le 1/K$. Then we have
    \begin{align}
        \| \log\bD -\log\tbD\| &\le 2K^3\eps_0 \label{eqn:logbD-continuity}\\
        \| \bH_\alpha - \tbH_\alpha\| &\le 3K^2\eps_0 \label{eqn:bH-continuity}
    \end{align}
\end{prop}

\begin{proof}
By Lemma \ref{lem:op-norm-bound}, the uniform bounds on measurement errors translate to the following error bounds on $C$ and $B_\alpha$:
\begin{align}
    \|C-\tilde C\| &\le r \eps_0\label{eq:CtildeCbound}\\
    \|B_\alpha-\tilde B_\alpha\|&\le r\eps_0
\end{align}
Since $\epsilon_0\le 1/K$, Lemma \ref{lem:op-norm-bound} gives
    \begin{equation}
      C\succeq \tilde C-\|C-\tilde C\| \succeq \frac{2r}{K}-r\eps_0\succeq \frac{r}{K}
    \end{equation}
    and so $\max(\|C^{-1}\|,\|\tilde{C}^{-1}\|)\le K/r$.
We also have $\|C\|, \|\tilde{C}\|\le r$ and $\|B_\alpha \|, \|\tilde{B}_{\alpha}\| \le 2r$ (using the assumptions (\ref{eqn:omegatilde-assum1}) and (\ref{eqn:omegatilde-assum3}) for the bounds on $\|\tilde{C}\|$ and $\|\tilde{B}_{\alpha}\|$).
Using Lemma~\ref{lem:matrix-continuity-bounds}, we have
\begin{align}
    \|\bD-\tilde\bD\|&\le\|C^{-1/2}-\tilde C^{-1/2}\|(\|C\|\|C^{-1/2}\|+\|\tilde C\|\|\tilde C^{-1/2}\|)+\|C-\tilde C\|\|C^{-1/2}\|\|\tilde C^{-1/2}\|\\
    &\le\|C-\tilde C\|\left(\frac{\|C\|\|C^{-1}\|\|\tilde C^{-1/2}\|+\|\tilde C\|\|\tilde C^{-1}\|\|C^{-1/2}\|}{\|C^{-1/2}\|^{-1}+\|\tilde C^{-1/2}\|^{-1}}+\|C^{-1/2}\|\|\tilde C^{-1/2}\|\right)\\
    &\le \left(K^2 + K\right)\eps_0\\
    &\le 2K^2\eps_0
\end{align}
Where in the last line we used the fact that $K\ge 2r\Tr(\tilde{C}^{-1})\ge 2r/\Tr(\tilde{C})=2\ge 1$. The bounds $r/K \preceq C \preceq r$ (resp. $r/K \preceq \tilde{C}\preceq r$) imply that $\|\bD^{-1}\|\le K$ (resp. $\|\tbD^{-1}\|\le K$), which gives (\ref{eqn:logbD-continuity}). Finally, a similar calculation gives (\ref{eqn:bH-continuity}):
\begin{align}
\|\bH_\alpha-\tilde\bH_\alpha\| &\le\|C-\tilde C\|\frac{\|B_\alpha\|\|C^{-1}\|\|\tilde C^{-1/2}\|+\|\tilde B_\alpha\|\|\tilde C^{-1}\|\|C^{-1/2}\|}{\|C^{-1/2}\|^{-1}+\|\tilde C^{-1/2}\|^{-1}}+\|B_\alpha-\tilde B_\alpha\|\|C^{-1/2}\|\|\tilde C^{-1/2}\|\\
&\le (2K^2+K)\eps_0\\
&\le 3K^2\eps_0.\label{eqn:Halpha-continuity}
\end{align}
\end{proof}
Theorem \ref{thm:a-posteriori-feasibility} then follows immediately from Proposition \ref{prop:constraint-continuity-2} and the matrix EEB inequality (Theorem \ref{thm:mEEB-ideal}).

\paragraph{Proof of Theorem \ref{thm:a-priori-feasibility}}\label{appendix:a-priori-feasibility-proof}
This will follow from a lower bound on the local marginals of Gibbs states that was first proven in \cite{Anshu2021}. We use a form of this bound that was given in \cite{bakshi2023}:
\begin{lemma}[{\cite[Corollary 2.14]{bakshi2023}}]\label{lem:marginalbound}
Let $E_\alpha$ be $k$-local Paulis such that the Hamiltonian is $k$-$\md$-low intersection and let $P_1,\ldots,P_r$ be distinct operators in $P_{k,\ell}$.
There exist constants $\mathcal{C}_{k,\mathfrak d,\ell}, \mathcal{D}_{k,\mathfrak d,\ell}$ depending only $k,\md \ell$ such that
    \begin{align}
        C \succeq \exp(-\mathcal{C}_{k,\mathfrak d,\ell}\beta-\mathcal{D}_{k,\mathfrak d,\ell})/r
    \end{align}
\end{lemma}
With this ingredient at hand we proceed to prove Theorem~\ref{thm:a-priori-feasibility}.

Set $\sigma := \exp(-\mathcal{C}_{k,\mathfrak d,\ell}\beta-\mathcal{D}_{k,\mathfrak d,\ell})/4r^2$ and suppose $\eps_0\le \sigma$. Then by Lemma \ref{lem:op-norm-bound} we have
\begin{align}
    \|\tilde{C}-C\| &\le r\eps_0\\
    &\le \exp(-\mathcal{C}_{k,\mathfrak d,\ell}\beta-\mathcal{D}_{k,\mathfrak d,\ell})/2r.
\end{align}
By Lemma \ref{lem:marginalbound} we get $\tilde{C}\succeq \exp(-\mathcal{C}_{k,\mathfrak d,\ell}\beta-\mathcal{D}_{k,\mathfrak d,\ell})/2r = 2r\sigma$, and so $K= 2r\|\tilde{C}^{-1}\| \le 1/\sigma$.
\qed

\section{Convergence proofs}
\label{sec:convergence}
In this section we take on the proof of the main convergence guarantee, Theorem \ref{thm:commuting-convergence}.
The proof will crucially use some basic concepts from modular theory, which we recall now.

\subsection{Modular theory}
Consider a quantum system described by a finite-dimensional Hilbert space $\mathcal{H}$. In this section we will, as a rule, use lowercase letters for elements of $\BH$ and uppercase letters for operators $\BH\to\BH$ (sometimes called superoperators). Define a \textit{state} as a complex-linear map $\omega:\BH \to \C$ satisfying $\omega(1)=1$ and $\omega(a^*a)\ge 0$ for every $a\in \mathcal{B}(H)$. These are in one-to-one correspondence with density matrices, i.e., positive-semidefinite trace-one operators $\rho$, via $\omega(a)=\Tr(\rho a)$ for any $a\in \mathcal{B}(H)$. A state $\omega$ is called \textit{faithful} if $\omega(a^*a)>0$ for every nonzero $a$, or equivalently, if its density matrix is positive-definite. This is the case for Gibbs states.
For a faithful state $\omega$, the inner product $\angles{a|b} := \omega(a^*b)$ is known as the \textit{Gelfand-Naimark-Segal (GNS)} inner product, and endows $\BH$ with the structure of a Hilbert space \cite{BR1}. We refer to an operator $a\in \BH$ as $|a\rangle$ when thought of a vector in this Hilbert space. 

The $*$-operation $S:|a\rangle \to |a^*\rangle$ is an antilinear involution on $\BH$. Consider its polar decomposition\footnote{Polar decomposition of antilinear operators is discussed in Appendix \ref{appendix:polar-decomp}}:
\begin{align}
    S = J\Delta^{1/2} = \Delta^{-1/2}J, 
\end{align} where $\Delta := S^\dagger S$ is positive-definite and $J := S\Delta^{-1/2}$ is anti-unitary. It is simple to check\footnote{Recall the adjoint of an antilinear operator $T$ is defined by the relation $\angles{u|T^\dagger v} := \overline{\angles{Tu|v}}$ for all vectors $u$ and $v$.} that $\angles{a|\Delta|b} = \omega(ba^*)$ for every $a,b\in \BH$ and that $\Delta$ acts on any $|a\rangle \in \BH$ by
\begin{align}
\Delta|a\rangle &= |\rho a\rho^{-1}\rangle \label{eqn:Delta-action}\\
\log(\Delta)|a\rangle &= |[\log(\rho),h]\rangle \label{eqn:logDelta-action}.
\end{align}
By Lemma \ref{lem:polar-decomp} in Appendix \ref{appendix:polar-decomp}, the operators $\Delta$ and $J$ satisfy:
\begin{align}
    J \log(\Delta)J &= -\log(\Delta).
\end{align}

If $h\in \BH$ is self-adjoint we denote by $H \in \mathcal{B}(\mathcal{H}_\omega)$ the operator $H:|a\rangle\mapsto |[h,a]\rangle$. We call it the GNS Hamiltonian corresponding to $h$.
Note that even though $h$ is hermitian, $H$ in general is not. In fact, $H$ is hermitian iff $h$ is a \textit{symmetry} of $\omega$:
\begin{lemma}\label{lem:J-odd-symmetry-exact}
    Let $h\in \BH$ be hermitian and let $H:\BH\to\BH$ be its GNS Hamiltonian. The following are equivalent:
    \begin{enumerate}[i)]
        \item $H$ is hermitian
        \item $J HJ=-H$
        \item $[h,\rho]= 0$
    \end{enumerate}
\end{lemma}
\begin{proof}
    $i)\iff iii)$. It is easy to check that $\angles{a|H-H^\dagger|a} = \omega([h,a^*a])$ for any $a\in \BH$. Thus $H^\dagger = H$ iff $\Tr(\rho[h,a^*a])= \Tr(a^*a[\rho,h]) = 0$ for all $a\in \BH$, which is equivalent to $[\rho,h]=0$.

    $ii)\iff iii)$. It is easy to check that $SHS=-H$. By (\ref{eqn:Delta-action}) we have $H + JHJ = H -\Delta^{1/2}H\Delta^{-1/2}$ which takes $|a\rangle \in \BH$ to $|[h-\rho^{1/2}h\rho^{-1/2},a]\rangle$. Thus $H + JHJ=0$ iff $h-\rho^{1/2}h\rho^{-1/2}$ is a multiple of the identity. But $\Tr(h-\rho^{1/2}h\rho^{-1/2})=0$ and so this can only happen when $[h,\rho^{1/2}]=0$, which is equivalent to $[h,\rho]=0$.
\end{proof}
The next lemma bounds the norm of the GNS Hamiltonian of a symmetry:
\begin{lemma}\label{lem:Hnorm}
    If $[\rho,h]=0$ then the GNS Hamiltonian $H$ of $h$ satisfies
    \begin{align}
        \|H\|_{gns} \le 2\|h\|.
    \end{align}
\end{lemma}
In the above, we write $\|H\|_{gns}$ to indicate that this refers to the norm of $H$, which is an operator on the GNS Hilbert space $\BH$, while $\|h\|$ refers to the norm of $h$, which is an operator on the physical Hilbert space $\mathcal{H}$.
\begin{proof}
    Denote the eigenbasis of $h$ (in $\mathcal{H}$) by $h\psi_i=e_i\psi_i$.
    Then, an eigenbasis of $H$ can be written as $|a_{ij}\rangle =|\psi_i\psi_j^*\rangle$, i.e., $H |a_{ij}\rangle =(e_i-e_j)|a_{ij}\rangle$. Since $[\rho,h]=0$, $H$ is hermitian, and so its operator norm can be bounded from its eigenvalues as $\|H\|\le\max_{i,j}|e_i-e_j|\le2\|h\|$.
\end{proof}
A fundamental result in modular theory relates the modular operator $\Delta$ to the GNS Hamiltonian $H$:
\begin{lemma}\label{lem:mEEB-ideal-ideal}
    Let $h \in \BH$ be hermitian. Then the following are equivalent:
    \begin{enumerate}[i)]
        \item $\log(\Delta) + H \succeq 0$ 
        \item $\log(\Delta) + H = 0$ 
        \item $\rho = e^{-h}/\Tr(e^{-h})$.
    \end{enumerate}
\end{lemma}
\begin{proof}
    $i)\iff ii)$. Since $\log(\Delta)$ is automatically Hermitian, $i)$ implies $H^\dagger = H$, and so $J HJ=-H$, and so
    \begin{align}
        0 &\preceq J (\log(\Delta) + H) J\\
        &= -\log(\Delta) - H
    \end{align}
    which implies $ii)$. The inverse implication is obvious.
    
     $ii)\iff iii)$. By (\ref{eqn:logDelta-action}), $ii)$ holds if and only if $\log(\rho)-h$ is a multiple of the identity, which is equivalent to $iii)$.
\end{proof}

\subsection{Restricted GNS space}
As we will show, our choice of perturbing operators selects a subspace of the GNS space, and the matrices appearing in the matrix EEB inequality are naturally interpreted as operators acting on this subspace. Let $P_1\hdots, P_r \in \BH$ be a set of selfadjoint operators and let $h'\in \BH$ be hermitian. Define the $r\times r$ matrices 
\begin{align}
    \bD &:= C^{-1/2}C^TC^{-1/2}\\
    \bH' &:= C^{-1/2}B'C^{-1/2}.
\end{align}
where 
$C_{ij}:=\omega(P_iP_j)$ and $B'_{ij} := \omega(P_i[h',P_j])$. 
Define $|a_i\rangle := \sum_{i=1}^r(C^{-1/2})_{ij}|P_j\rangle$ for $i=1,\hdots r$. It is easy to check that the vectors $|a_i\rangle$ form a basis of $\mathcal{P}:= \spn\{P_1,\hdots, P_r\}$ that is \textit{orthonormal} in the GNS inner product, and that we have 
\begin{align}
    \bD_{ij} &:= \angles{a_i|\Delta|a_j}\\
    \bH'_{ij} &:= \angles{a_i|H'|a_j}
\end{align}
where $H'$ is the GNS Hamiltonian of $h'$. Thus, with some abuse of notation we may write
\begin{align}
    \bD = Q\Delta Q\\
    \bH' = QH'Q.
\end{align} where $Q:\BH\to\BH$ is the orthogonal projection onto $\mathcal{P}$. We also define the restricted $*$-operation $\bS$ as
\begin{align}
    \bS := QSQ,
\end{align}
which can equivalently be defined in terms of the correlation matrix $C$ as $\bS = C^{-1/2}\left(\overline{C}\right)^{-1/2}$, where $\overline{C}$ is the complex conjugate of $C$. Finally, we will use an operator $\bJ$, which is defined as
\begin{align}
    \bJ := \bS\bD^{-1/2}.
\end{align}
It is important to stress that unlike $\bD$, $\bH$, and $\bS$, it turns out that the operator $\bJ$ is \textit{not} simply the restriction of $J$ to the span of $P_1,\hdots, P_r$, ie. $\bJ\neq QJQ$. Instead, one can check that $\bD = \bS^\dagger \bS$ and it follows that $\bJ$ is the anti-unitary part in the polar decomposition of $\bS$. Since $\bS$ is an antilinear involution just like $S$, Lemma \ref{lem:polar-decomp} gives
\begin{align}
    \bJ^2 &=\boldsymbol{1}, &&\bJ^\dagger = \bJ\\
    \bJ \bD \bJ &= \bD^{-1},  &&\bJ \bD^{1/2} \bJ = \bD^{-1/2}\\
    \bJ\log(\bD)\bJ &= -\log(\bD)
\end{align}
\subsection{Proof of Theorem \ref{thm:commuting-convergence}}\label{sec:commuting-convergence}
In this section, we prove our main convergence theorem for commuting Hamiltonians. Fix $k,\md >0$ and suppose $E_1,\hdots E_m$ are $k$-supported Paulis with $\md$-low-intersection. Fix $\ell>0$, and let $P_1,\hdots P_r := \mathcal{P}_{k,\ell}$. Suppose $\rho$ is the Gibbs state of a Hamiltonian $h=\sum_\alpha \lambda_\alpha E_\alpha$ for some parameters $\lambda\in \R^m$ with $\max_{\alpha=1,\hdots, m}|\lambda_\alpha|\le \beta$. We do not assume that $h$ is commuting throughout. Instead, we will state explicitly in the statement of each Lemma/Proposition/Theorem if we assume that $h$ is commuting. 
Throughout this section $\lambda'\in \R^m$ will refer to an arbitray vector of parameters. We will write $h'=\sum_\alpha \lambda_\alpha'E_\alpha$ and $H'$ for its GNS Hamiltonian. Finally, we will use the shorthand notations
\[
\bH'=\sum_{\alpha=1}^m\lambda_\alpha' \bH_\alpha\quad\textrm{ and }\quad\bH=\sum_{\alpha=1}^m\lambda_\alpha \bH_\alpha,
\]
where $\bH_{\alpha}$ are the $r\times r$ matrices defined in (\ref{eqn:bH-alpha-definition}).

Let us now give an overview of the proof. The bulk of the proof will boil down to establishing convergence of the \textit{relaxed} EEB constraints, for the \textit{exact} expectation values, ie. with no measurement noise.
For each Hamiltonian parameter, the proof first identifies local witnesses of deviations: local operators $a$ such that $\angles{a|\bH'-\bH|a}$ detects the difference in Hamiltonian parameters $|\lambda_\alpha - \lambda_\alpha'|$ for a given $\alpha=1,\hdots, m$.
While Lemma \ref{lem:J-odd-symmetry-exact} above shows that $H'^\dagger = H'$ is equivalent to $JH'J=-H'$, we only enforce $\bH^\dagger = \bH$ up to some error, and in Lemma~\ref{lem:J-odd-approx} and Corollary~\ref{cor:proj-J-odd-approx} we show that this gives an approximate version of $\bJ\bH'\bJ = -\bH'$.
A crucial ingredient is the fact that the time-evolution, and thereby the $J$ operation, for commuting Hamiltonians maps local operators to exactly local operators, preventing them from leaving the restricted GNS space, see Lemma~\ref{lem:commuting}.
Using an antisymmetry argument based on the one in Lemma \ref{lem:mEEB-ideal-ideal} find the desired bound on matrix elements of $H-H'$ (Proposition~\ref{prop:noiseless-convergence-2}), which conclude the proof for the noiseless case. Finally, we prove Theorem \ref{thm:commuting-convergence} from the noiseless case via the continuity bounds proved in Section \ref{sec:noise-stability}.

We begin with a bound on $r$, the number of perturbing operators, which has been shown in \cite{bakshi2023}.
\begin{lemma}[{\cite[Corollary 2.20]{bakshi2023}}]\label{lem:r-count}
    The size of the set $\mP_{k,\ell}$ is bounded by $m\md^\ell10^{k\ell}$.
\end{lemma}

Since our constraints work with the GNS-Hamiltonian we need to relate its matrix elements to the coefficients of its parent Hamiltonian to witness large errors in the output parameters.
The following Lemma achieves this and furthermore singles out one coefficient, which is needed to achieve bounds uniformly.
\begin{lemma}[Local identifiability of Hamiltonian terms]\label{lem:ham-term-detection}
Suppose $\ell \ge 1+\md$. For each Hamiltonian coefficient $\lambda_\alpha$, there are $k$-$\mathfrak (1+\md)$-$\mG$-supported operators $a_1,a_2,a_3,a_4$ with $\|a_i\|\le e^{\mathcal{O}_{k,\md}(\beta)}$ and
\[
|\lambda_\alpha-\lambda'_\alpha|=\left|\sum_{i=1}^4\langle a_i|\bH-\bH'|a_i\rangle\right|.
\]
\end{lemma}
\begin{proof}
Recall from \cite[Lemma 9.8]{bakshi2023} that there are $a',b$ such that $\|a\|=\|b\|=1$, $\supp(a')=\supp(b)=\supp(E_\alpha)$\footnote{The statement of the Lemma in the reference is slightly weaker, but the argument of the proof directly yields this stronger constraint on the supports.} and
\[
\left|\frac12\tr([h-h',b]a')\right|=|\lambda_\alpha-\lambda'_\alpha|
\]
by denoting by $\overline{\rho}$ the marginal of the Gibbs state $\omega$ on the union of supports of $[h-h',b]$ and $a'$, by $D_{\overline{\rho}}$ the dimension of this support, and defining $a=a'(\overline{\rho}D_{\overline{\rho}})^{-1}$ we have
\begin{align}
|\lambda_\alpha-\lambda'_a|&=\left|\frac12\tr([h-h',b]a\overline{\rho}D_{\overline{\rho}})\right|\\
&=\frac12|\omega([h-h',b]a)|\\
&=\frac12|\langle a|H-H'|b\rangle|\\
&=\frac12|\langle a|\bH-\bH'|b\rangle|,\label{eqn:bound-pre-polarization}
\end{align}
where the last line is because $\ell \ge 1+\md$.
Then by the ``no small local marginals" result \cite[Corollary 2.14]{bakshi2023} or concretely the formulation in terms of density matrices \cite[Lemma 3.8]{fawzi2023}, we have $\|a\|\le\exp(\mathcal O_{k,\md}(\beta))$.
Note that $a$ and $b$ are both supported on the union of $\supp(E_\alpha)$ and all $\supp(E_{\alpha'})$ intersecting with $\supp(E_\alpha)$. In particular, any linear combination of $a$ and $b$ is $k$-$(\md+1)$-$\mG$-supported.

Applying the polarization identity to (\ref{eqn:bound-pre-polarization}) gives
    \begin{align}
        |\lambda_\alpha - \lambda_\alpha'| &= \frac{1}{2}\left|\angles{a|\bH - \bH'|b}\right|\\
        &= \left|\frac{1}{8}\sum_{n=1}^4\angles{a+i^nb|\bH - \bH'|a+i^nb}\right|.
    \end{align}
\end{proof}

The hermiticity of the GNS-Hamiltonian is equivalent to a stationary condition of the state (see Lemma \ref{lem:J-odd-symmetry-exact}).
As part of our constraint system we enforce approximate hermiticity for its measured version.
The following Lemma leverages this constraint to prove that $h'$ approximately commutes with $h$ and with $\rho^{1/2}$.
\begin{lemma}\label{lem:approximate-stationarity}
    Suppose $\ell\ge 3$ and suppose $h' = \sum_{\alpha=1}^m\lambda_\alpha'E_\alpha$ satisfies $\|\bH' -(\bH')^\dagger\|\le\mu$. Then we have
    \begin{align}
        \|H|h'\rangle\|_{gns} &\le m\beta\mu \label{eqn:approximate-stationarity-1}\\
        \|\Delta^{1/2}|h'\rangle - |h'\rangle\|_{gns} &\le \frac{1}{2}(m\beta)^{1/2}\mu\label{eqn:approximate-stationarity-2}
    \end{align}
\end{lemma}
\begin{proof}
    It is easy to check that for $k$-$\ell$-$\mG$-local operators $a,b$ we have \begin{align}
        |\angles{a|H'-H'^\dagger|b}| &= |\angles{a|\bH' - (\bH')^\dagger|b}|\\
        &\le \mu \||a\rangle\|_{gns}\||b\rangle\|_{gns}
    \end{align}
    Since $\ell \ge 3$, setting $|a\rangle = |1\rangle$ and $|b\rangle = |[h,h']\rangle = H|h'\rangle$ we have 
    \begin{align}
        |\angles{1|H'H|h'}| &= |\angles{1|(H'-H'^\dagger)H|h'}|\\
        &\le \mu \|H|h'\rangle\|_{gns} 
    \end{align}
    since $H'|1\rangle=0$. An elementary calculation using the fact that $\omega([h,a]) = 0$ for all $a\in \BH$ shows that
    $\angles{1|H'H|h'} = 2\angles{h'|H|h'}$ and so
    \begin{align}
        \angles{h'|H|h'}\le\frac12\mu \|H|h'\rangle\|_{gns}.
    \end{align}
    Using the above inequality with Lemma \ref{lem:Hnorm} and using the fact that $\|h\|\le m\beta$, we have
    \begin{align}
        \|H|h'\rangle\|^2_{gns} &= \angles{h'|H^2|h'}\\
        &\le 2m\beta\angles{h'|H|h'}\\
        &\le m\beta\mu \|H|h'\rangle\|_{gns},
    \end{align}
    which proves (\ref{eqn:approximate-stationarity-1}). For (\ref{eqn:approximate-stationarity-2}), using $\Delta = e^{-H}$ we write:
    \begin{align}
         \|\Delta^{1/2}|h'\rangle - |h'\rangle\|^2_{gns} &= \angles{h'|(e^{-H/2}-1)^2|h'}. \label{eqn:approximate-stationarity-3}
    \end{align}
    For any $x\in \R$ we have
    \begin{align}
        (e^{x/2}-1)^2= \left(\frac{\tanh(x/4)}{x}\right) x(e^x-1) \le \frac{x(e^x-1)}{4}
    \end{align}
    replacing $x$ with $-H$ using the functional calculus we may continue (\ref{eqn:approximate-stationarity-3}) as follows:
    \begin{align}
        \|\Delta^{1/2}|h'\rangle - |h'\rangle\|^2_{gns} &\le \frac{1}{4}\angles{h'|H(1-e^{-H})|h'}\\
        &= \frac{1}{4}\left(\angles{h'|H|h'} - \angles{h'|H\Delta|h'}\right)\\
        &= \frac{1}{4}\left(\angles{h'|H|h'} - \angles{h'|[h,h']^*}\right)\\
        &= \frac{1}{2}\angles{h'|H|h'}\\
        &\le\frac14 m\beta\mu^2.
    \end{align}
\end{proof}

\begin{lemma}\label{lem:J-odd-approx}
    Suppose $\ell\ge 3$, and let $h'$ be a Hamiltonian with $\|\bH'-(\bH')^\dagger\| \le \mu$ for some $\mu>0$. Suppose $a$ is an operator such that $\|\rho a \rho^{-1}\| \le \mD\|a\|$ for some $\mD>0$. Then 
    \begin{align}
        |\bra{a}JH'J+H'\ket{a}|\le\frac12\|a\|^2(\mD+1)(m\beta)^{1/2}\mu \label{eqn:J-odd-approx}
    \end{align}
\end{lemma}
\begin{proof}
    It is easy to check that $SH'S = -H'$. Indeed, for any $a\in \BH$ we have
    \begin{align}
        SHS|a\rangle = |[h,a^*]^*\rangle = -|[h,a]\rangle.
    \end{align}
    Therefore,
    \begin{align*}
    \bra{a}JH'J\ket{a}&=\bra{a}\Delta^{1/2}SH'S\Delta^{-1/2}\ket{a}\\
        &=-\bra{a}\Delta^{1/2}H'\Delta^{-1/2}\ket{a}
    \end{align*}
    and so it suffices to bound the magnitude of $\angles{a|\Delta^{1/2}H'\Delta^{-1/2} -H'|a}$. Writing this quantity in terms of operators on the physical Hilbert space and using the identities $\Delta^{1/2}|a\rangle = |\rho^{1/2}a\rho^{-1/2}\rangle$,  $\omega(ab)=\omega(b(\rho a \rho^{-1})) = \omega((\rho^{-1} b \rho) a)$, and $\omega(\rho^{-1/2}a\rho^{1/2}) = \omega(a)$, we have
    \begin{align}
        \angles{a|\Delta^{1/2}H'\Delta^{-1/2} -H'|a}
        &= \omega((\rho^{1/2}a\rho^{-1/2})^*[h',\rho^{-1/2}a\rho^{1/2}]) - \omega(a^*[h',a])\\ 
        &= \omega(\rho^{-1/2}a^*\rho^{1/2}h'\rho^{-1/2}a\rho^{1/2}) - \omega(\rho^{-1/2}a^*a\rho^{1/2}h') - \omega(a^*[h',a])\\
        &= \omega(\rho^{-1}a\rho a^*\rho^{1/2}h'\rho^{-1/2}) - \omega(\rho^{-1/2}a^*a\rho^{1/2}h')
        -\omega(\rho^{-1}a\rho a^*h') + \omega(a^*ah')\\
        &=\omega((\rho^{-1}a\rho a^* - a^*a)(\rho^{1/2}h'\rho^{-1/2}-h'))\\
        &= \angles{a\rho a^* \rho^{-1} - a^* a|\rho^{1/2}h'\rho^{-1/2} - h'}
        \label{eqn:cs}
    \end{align}
    By the Cauchy-Schwartz inequality and Lemma \ref{lem:approximate-stationarity} we get
    \begin{align}
        |\angles{a|\Delta^{1/2}H'\Delta^{-1/2} -H'|a}|
        &\le \frac12(m\beta)^{1/2}\mu \||a\rho a^* \rho^{-1} - a^* a\rangle \|_{gns}\\
        &\le \frac12(m\beta)^{1/2}\mu\left( \|a\rho a^* \rho^{-1}\| + \|a^* a\| \|\right)\\
        &\le \frac12\|a\|^2(\mD+1)(m\beta)^{1/2}\mu.
    \end{align}
\end{proof}
\begin{lemma}\label{commuting-imag-timeevol-normbound}
    Suppose $h$ is commuting as in Definition \ref{def:commuting}. Then for any $k$-$\ell'$-$\mG$-supported operator $a$ we have
    \begin{align}
        \|\rho a \rho^{-1}\| \le e^{2\mathcal{C}\beta (1+\md)\ell'}\|a\|.
    \end{align}
\end{lemma}

\begin{proof}
    Since $a$ is $k$-$\ell'$-$\mG$-supported there is a set $L\subset \{1,\hdots, m\}$ with $|L|=\ell'$ such that $\supp(a) \subset \bigcup_{\alpha\in L}\supp(E_\alpha)$. Let $\tilde{L}$ be the set of $\alpha \in \{1,\hdots, m\}$ for which there is an $\alpha' \in L$ with $\supp(E_\alpha)\cap \supp(E_{\alpha'})\neq \varnothing\}$.
    Then $|\tilde{L}|\le (1+\md)\ell'$.
    Let $\tilde{h} = \sum_{\alpha \in \tilde{L}} \nu_\alpha F_\alpha$. Then $\|\tilde{h}\|\le \mathcal{C}\beta(\md+1)\ell'$ and so
    \begin{align}
        \|\rho a\rho^{-1}\| &= \|e^{-\tilde{h}}ae^{\tilde{h}}\|\\
        &\le \|a\|e^{2\mathcal{C}\beta(\md+1)\ell'}.
    \end{align}
\end{proof}

The following Lemma relates $\Delta$ and $J$ to their local counterparts in a strong way.
It is based on the commuting Hamiltonian assumption, exploiting the fact that the complex time-evolution operator of commuting Hamiltonians maps local observables to strictly local observables.
\begin{lemma}\label{lem:commuting}
Suppose $h$ is commuting.
\begin{enumerate}
    \item For any $\ell'> 0$, any $k$-$\ell'$-$\mG$-supported operator $a$, and any nonnegative integer $p$, the operator $\Delta^p\ket{a}$ is $k$-$(1+\md)\ell'$-$\mG$-supported.
    \item For any $\ell'\le\ell/(1+\md)$, any $k$-$\ell'$-$\mG$-supported $a$, and any nonnegative integer $p$ we have
    \begin{align}
        \Delta^p|a\rangle = \bD^p|a\rangle.
    \end{align}
    \item For any $\ell'\le\ell/(1+\md)$ and any $k$-$\ell'$-$\mG$-supported $a$ we have
    \begin{align}
        \Delta^{1/2}|a\rangle &= \bD^{1/2}|a\rangle\\
        J|a\rangle &= \bJ|a\rangle
    \end{align}
\end{enumerate}
\end{lemma}
\begin{proof}

$1.$\\
We have to consider the operator
\[
e^{-p h}ae^{p h }.
\]
We can define the set $A$ of indices such that $[F_\alpha,a]=0$ for all $\alpha\in A$ such that
\begin{align*}
e^{-p h}ae^{p h }&=e^{-p\sum_{\alpha\in A^c}\nu_\alpha F_\alpha}a e^{-p\sum_{\alpha\in A}\nu_\alpha F_\alpha}e^{p\sum_{\alpha\in A}\nu_\alpha F_\alpha}e^{p\sum_{\alpha\in A^c}\nu_\alpha F_\alpha}\\
&=e^{-p\sum_{\alpha\in A^c}\nu_\alpha F_\alpha}a e^{p\sum_{\alpha\in A^c}\nu_\alpha F_\alpha}.
\end{align*}
By the definition of $k$-$\ell'$-$\mG$-locality, the operator $a$ is supported on $\bigcup_{\alpha\in S}E_\alpha$ for some connected $S$ with $|S|\le\ell'$.
Since all terms $F_\alpha$ for $\alpha\in A^c$ have overlapping support with $\supp(a)$, all terms in the above equation have support in $\bigcup\{\supp(E_\alpha)| \exists\alpha'\textrm{ s.t. }\supp(E_\alpha)\cap\supp(E_{\alpha'})\neq\emptyset\}$, which is the union of at most $(1+\md)\cdot\ell'$ connected supports and thereby $k$-$(1+\md)\ell'$-$\mG$-local. 

$2.$\\
The case $p=0$ is trivial and we prove the cases $p>0$ by induction.
Write $Q:\BH\to\BH$ for the orthogonal (in the GNS inner product) projection onto $\spn\{P_1,\hdots, P_r\} \subset \BH$. Notice that for any $\ell$-supported operator $b\in \BH$ we have $Q|b\rangle = |b\rangle$. The induction then follows from
\begin{align}
    \bD^{p+1}|a\rangle &= Q\Delta Q \bD^{p}|a\rangle\\
    &= Q\Delta Q \Delta^{p}|a\rangle\\
    &= Q \Delta^{p+1}|a\rangle\\
    &= \Delta^{p+1}|a\rangle.
\end{align}
The second line is by the inductive hypothesis, and the third and fourth lines are by part $1$.

$3.$\\
Let $x_0=\max\{\|\Delta\|,\|\bD\|\}$ and $x^{1/2} = \sum_{k=0}^\infty a_k(x_0-x)^{k}$ be the Taylor series for the square root.
By part 2, we have
\begin{align}
    \Delta^{1/2}|a\rangle &= \sum_{k=0}^\infty a_k(x_0-\Delta)^{k}|a\rangle \\
    &= \sum_{k=0}^\infty a_k(x_0-\bD)^{k}|a\rangle\\
    &= \bD^{1/2}|a\rangle.
\end{align}
The second statement then follows from $\bJ = \bD^{1/2}\bS$.
\end{proof}
\begin{corollary}\label{cor:proj-J-odd-approx}
    Suppose $\ell \ge 3$ and that $h$ is commuting. Let $h'$ be a Hamiltonian with $\|\bH'-\bH'^\dagger\|\le \mu$ for some $\mu \ge 0$. Let $a$ be $k$-$\ell'$-$\mG$-supported for some $\ell'\le (\ell-1)/(1+\md)$. Then,
    \begin{align}
        |\bra{a}\bJ \boldsymbol{H'}\bJ +\boldsymbol{H'}\ket{a}|\le \frac12\|a\|^2(e^{2\mathcal{C}\beta(\md+1)\ell'}+1)(m\beta)^{1/2}\mu.
    \end{align}
\end{corollary}
\begin{proof}
    Note that by Lemma~\ref{lem:commuting}, $\bJ\ket{a}$ is $k$-$\ell'(1+\md)$-$\mG$-supported.
    Using that $H$ increases the locality by one and Lemmas \ref{lem:J-odd-approx} and \ref{commuting-imag-timeevol-normbound}, we have
    \begin{align}
        |\angles{a|\bJ \bH'\bJ +\boldsymbol{H'}|a} &= |\angles{a|\bJ H'\bJ +H'|a}\\
        &= |\angles{a|J H'J+H'|a}\\
        &\le\frac12\|a\|^2(e^{2\mathcal{C}\beta(\md+1)\ell'}+1)(m\beta)^{1/2}\mu.
    \end{align}
\end{proof}
\begin{prop}\label{prop:noiseless-convergence-2}
    Suppose for some $\mu_1,\mu_2 \ge 0$ that $h'$ satisfies
    \begin{align}
        \log(\bD)+\sum_\alpha \lambda_\alpha'(\boldsymbol{H}_\alpha+\bH_\alpha^\dagger)/2&\ge-\mu_1 \label{eqn:convergence-1}\\
        \pm i\sum_\alpha\lambda_\alpha'(\bH_\alpha-\bH_\alpha^\dagger)/2&\le \mu_2 \label{eqn:convergence-2}
    \end{align}
    and let $a\in \BH$ be an operator with
    \begin{align}
        \angles{a|\bH + \bJ \bH \bJ|a} &= 0 \label{eqn:noiseless-convergence-a-assumption-1}\\
        \Re\left(\angles{a|\bH' + \bJ \bH' \bJ|a}\right) &\le \delta \|a\|^2 \label{eqn:noiseless-convergence-a-assumption-2}
    \end{align}
    for some $\delta\ge 0$. Then
    \begin{align}
        |\angles{a|\bH' - \bH|a}|  \le (\mu_1+\mu_2+\delta)\|a\|^2. 
    \end{align}
\end{prop}
\begin{proof}
    Let $\bH'_{\pm} := (\bH'\pm \bH'^\dagger)/2$, $\bH_{\pm} := (\bH\pm \bH^\dagger)/2$. Applying antisymmetry of $\log\bD$ under conjugation by $\bJ $, $\bJ =\bJ ^\dagger$ and $\bJ ^2=1$ to (\ref{eqn:convergence-1}) we get
    \begin{align}
        \log\bD + \bH_{+}' &\ge -\mu_1 \label{eqn:convergence-matrix-eeb'}\\
        -\log\bD + \bJ \bH_{+}'\bJ  &\ge -\mu_1. \label{eqn:convergence-matrix-eeb-conjugated'}
    \end{align}
    At the same time, by Theorem~\ref{thm:mEEB-ideal}, we have $\bH_{+}=\bH$ and 
    \begin{align}
        \log\bD + \bH &\ge 0 \label{eqn:convergence-matrix-eeb}\\
        -\log\bD + \bJ \bH \bJ  &\ge 0. \label{eqn:convergence-matrix-eeb-conjugated}
    \end{align}
    Using (\ref{eqn:convergence-matrix-eeb-conjugated'}) and (\ref{eqn:convergence-matrix-eeb})
    we have
    \begin{align}
        \bH_+' - \bH &\le \bH_+' + \log\bD\\
        &= \bH_+' +\bJ\bH_+'\bJ - \bJ\bH_+'\bJ + \log\bD\\
        &\le \bH_+' + \bJ\bH_+'\bJ+\mu_1.
    \end{align}
    On the other hand, using (\ref{eqn:convergence-matrix-eeb'}) and (\ref{eqn:convergence-matrix-eeb-conjugated}) we have 
    \begin{align}
        \bH_+'-\bH &\ge -\log\bD -\mu - \bH\\
        &=  -\log\bD -\mu_1 - \bH - \bJ \bH\bJ + \bJ\bH\bJ\\
        &\ge -\mu_1 - (\bH + \bJ\bH\bJ).
    \end{align}
    Thus, for any operator $a$ we have
    \begin{align}
        -\mu_1\angles{a|a} -\angles{a|\bH+\bJ\bH\bJ|a} \le \angles{a|\bH_+'-\bH|a} \le \angles{a|\bH_+'+\bJ\bH_+'\bJ|a}+\mu_1\angles{a|a}.
    \end{align}
    Since $\angles{a|\bH_+' + \bJ \bH_+' \bJ|a} = \Re(\angles{a|\bH' + \bJ \bH' \bJ|a})$, if $a$ further satisfies (\ref{eqn:noiseless-convergence-a-assumption-1}) and (\ref{eqn:noiseless-convergence-a-assumption-2}) then we get
    \begin{align}
        |\Re(\angles{a|\bH' - \bH|a})| &= |\angles{a|\bH_+' - \bH|a}| \\
        &\le \mu_1\angles{a|a} + \delta|a\|^2\\
        &\le (\mu_1+\delta)\|a\|^2.
    \end{align}
    On the other hand
    \begin{align}
        |\Im(\angles{a|\bH' - \bH|a})| &= |\angles{a|\bH_-'|a}| \le\|a\|^2 \mu_2,
    \end{align}
    and putting this bound together with the previous one proves the Proposition.
\end{proof}
We are now ready to prove Theorem~\ref{thm:commuting-convergence}, which we restate for convenience:
\begin{theorem*}[A priori convergence in the commuting case]
    Suppose $h$ is commuting and that $\{P_1,\hdots, P_r\} = \mP_{k,\ell}$ for $\ell = \max(3,1+(1+\md)^2)$. There is an error threshold
    \begin{align}
        \tau &= m^{-6}e^{-\mathcal{O}_{k,\md,\mC}(\beta)}
    \end{align}
    (where $\mC$ is the constant from Definition \ref{def:commuting}) such that if $\epsilon_0 \le \tau$ then for any $\lambda'\in \R^m$ satisfying
    \begin{align}
         \log(\tbD) + \sum_{\alpha=1}^m\lambda'_\alpha(\tbH_{\alpha} + \tbH_{\alpha}^\dagger)/2  &\succeq -\mu_1, \label{linear-constraint-noisy-insec4} \\ 
        \pm i\sum_{\alpha=1}^m\lambda'_\alpha(\tbH_{\alpha} - \tbH_{\alpha}^\dagger)/2 &\preceq \mu_2. \label{SDP-constraint-noisy-insec4}
        \end{align} we have
    \begin{align}
        \sup_{\alpha=1,\hdots m}|\lambda_\alpha'-\lambda_\alpha| \le e^{\mathcal{O}_{k,\md,\mC}(\beta)}\left(\mu_1+m^{1/2}\mu_2\right) + \eps_0/\tau.
    \end{align}
\end{theorem*}

\begin{proof}
    Suppose $\lambda'\in \R^m$ satisfies (\ref{linear-constraint-noisy-insec4}) and (\ref{SDP-constraint-noisy-insec4}). By Theorem~\ref{thm:a-priori-feasibility}, there are constants $\mathcal{D},\mathcal{E}\ge 0$ depending only on $k$ and $\md$ such that setting
    \begin{align}
        \sigma := m^{-2}e^{-\mathcal{D}\beta - \mathcal{E}}
    \end{align}
    and assuming $\eps_0\le \sigma$, we have $K\le 1/\sigma = e^{\mathcal{O}_{k,\md}(\beta)}$. Applying the continuity bounds in Proposition \ref{prop:constraint-continuity-2} to (\ref{linear-constraint-noisy-insec4}) gives
    \begin{align}
        \log(\bD)+\sum_\alpha \lambda_\alpha'(\boldsymbol{H}_\alpha+\bH_\alpha^\dagger)/2 &\succeq -\mu_1 - (2K^3 + 3m\beta'K^2)\eps_0 \\
        &\succeq -\mu_1-e^{3\mathcal{D}\beta+3\mathcal{E}}\left( 2m^6 + 3m^5\beta'\right)\eps_0, \label{eqn:mu_1-tilde}
\end{align}
    where $\beta' = \max_{\alpha=1,\hdots m}|\lambda_\alpha'|$. Doing so to (\ref{SDP-constraint-noisy-insec4}) gives
\begin{align}
        \pm i\sum_\alpha\lambda_\alpha'(\bH_\alpha-\bH_\alpha^\dagger)/2 &\preceq \mu_2 + 3m\beta'K^2\eps_0\\
        &\preceq \mu_2 + 3m^5\beta'e^{2\mathcal{D}\beta+2\mathcal{E}}\eps_0.\label{eqn:mu_2-tilde}
\end{align}
Let $a$ be any $(k,1+\md)$-$\mG$-local operator. By Lemma \ref{lem:commuting} we have $\angles{a|\bH + \bJ\bH\bJ|a} = \angles{a|J + JHJ|a} = 0$. Applying this fact, together with the bound from Corollary \ref{cor:proj-J-odd-approx} and the inequalities (\ref{eqn:mu_1-tilde}) and (\ref{eqn:mu_2-tilde}), to Proposition \ref{prop:noiseless-convergence-2}, we have
\begin{align}
    |\angles{a|\bH - \bH'|a}| \le e^{\mathcal{O}_{k,\md,\mC}(\beta)}(\mu_1+m^{1/2}\mu_2+(m^6+m^{5.5}\beta')\eps_0),\label{eqn:lambadbound-in-proof-0}
\end{align}
and so by Lemma \ref{lem:ham-term-detection} we get
\begin{align}
    \max_{\alpha = 1,\hdots, m}|\lambda'_{\alpha} - \lambda_\alpha|
    &\le e^{\mathcal{O}_{k,\md,\mC}(\beta)}(\mu_1+m^{1/2}\mu_2+(m^6+m^{5.5}\beta')\eps_0) .\label{eqn:lambadbound-in-proof-1}
\end{align}
    The locality constraint $\ell=1+(1+\md)^2$ follows from combining the choices in Lemma~\ref{lem:ham-term-detection}, Lemma~\ref{lem:commuting}, and Corollary~\ref{cor:proj-J-odd-approx}, which are $\ell\ge3$, $\md+1=\ell' \le \ell/(1+\md)$ and $\md+1=\ell'\le(\ell-1)/(1+\md)$, respectively.
    This is essentially the result we need but contains a parameter $\beta'$.
    While it is possible to simply set explicit bounds on the parameter domain during the optimization to bound $\beta'$, we show in the following that this is not needed.
The previous expression can be summarized as
\begin{align}
    \max_{\alpha = 1,\hdots, m}|\lambda'_{\alpha} - \lambda_\alpha|
        &\le\gamma+\eps_0/\tau+\delta\beta'\eps_0
\end{align}
by defining constants that can be chosen to satisfy
\begin{align}
\gamma&\le e^{\mathcal{O}_{k,\md,\mC}(\beta)}(\mu_1+m^{1/2}\mu_2)\\
1/\tau&\le e^{\mathcal{O}_{k,\md,\mC}(\beta)}m^6\\
\delta&\le e^{\mathcal{O}_{k,\md,\mC}(\beta)}m^{5.5}
\end{align}
We can upper bound $\beta'$ as
\begin{equation}
    \beta'\le\beta+\max_{\alpha = 1,\hdots, m}|\lambda'_{\alpha} - \lambda_\alpha|\le\beta+\gamma+\eps_0/\tau+\delta\beta'\eps_0
\end{equation}
so by requiring $\eps_0\le1/2\delta$ we have
\begin{equation}
    \beta'\le2(\beta+\gamma+\eps_0/\tau).
\end{equation}
Plugging the bound for $\beta'$ back into the error estimates we obtain
\begin{align}
    \max_{\alpha = 1,\hdots, m}|\lambda'_{\alpha} - \lambda_\alpha|
        &\le2\gamma+2\eps_0/\tau+\delta\beta\eps_0.
\end{align}
The proof follows from recalling the conditions $\eps_0\le\sigma$ and $\eps_0\le1/2\delta$ and collecting the worst case estimates of all constants above.

\end{proof}

\printbibliography
\appendix
\section{Antilinear operators}\label{appendix:polar-decomp}
Let $\mathcal{K}$ be a finite-dimensional complex Hilbert space. A map $T:\mathcal{K}\to\mathcal{K}$ is called \textit{antilinear} if it is linear over $\R$ and satisfies $T(\lambda v) = \overline{\lambda}Tv$ for every $v\in \mathcal{K}$ and $\lambda\in \C$. Equivalently, $\mathcal{K}$ can be viewed as a real vector space (of double the dimension) and $T$ is a $\R$-linear operator on this real vector space that anticommutes with the $\R$-linear operator $v\mapsto \sqrt{-1}v$. The composition of a linear operator with an antilinear operator (in either order) is antilinear, and the composition of two antilinear operators is linear. If $T$ is antilinear then its adjoint is defined by the relation:
\begin{align}
    \angles{u|T^\dagger v} = \overline{\angles{Tu|v}} \hspace{5mm} \text{for all $u,v\in \mathcal{K}$}.
\end{align}
Note the complex conjugation, which is absent from the definition for complex-linear operators. An antilinear operator $U$ is called \textit{anti-unitary} if $U^\dagger U = 1$. One defines the polar decomposition of an antilinear operator $T$ in a way analogous to linear operators: the operators $U$ and $P$ form a polar decomposition of $T$ iff $U$ is anti-unitary, $P \succeq 0$ is linear, and $T = UP$. If $T$ is invertible then $U$ and $P$ are uniquely defined as $P := \sqrt{T^\dagger T}$ and $U = TP^{-1}$.

We call an operator $T$ an \textit{involution} if $T^2=1$. The following is a structure theorem for the the polar decomposition of an antilinear involution:
\begin{lemma}\label{lem:polar-decomp}
     Let $S:\mathcal{K}\to\mathcal{K}$ be an antilinear involution on $\mathcal{K}$. Define $\Delta := S^\dagger S$ and $J:= S\Delta^{-1/2}$ so that $S=J\Delta^{1/2}$ is the polar decomposition of $S$. Then we have
    \begin{enumerate}
    \item $S\log\Delta S = -\log \Delta$ and $S\Delta^{p}S = \Delta^{-p}$ for any $p\in \R$,
    \item $J = \Delta^{1/2}S$ and $J^\dagger = J$ and $J^2=1$,
    \item $J\log\Delta J = -\log \Delta$ and $J\Delta^{p}J = \Delta^{-p}$ for any $p\in \R$.
\end{enumerate}
\end{lemma}
\begin{proof}
    $1.$ Since $S^2=(S^\dagger)^2 =1$, we have $S\Delta S\Delta = SS^\dagger S S S^\dagger S = 1$, and so $S\Delta S = \Delta^{-1}$. From here we can write
    \begin{align*}
        e^{-\log\Delta} &= \Delta^{-1}\\
        &= S\Delta S\\
        &= S\sum_{k\ge 0}\frac{\log\Delta^k}{k!}S\\
        &= \sum_{k\ge 0}\frac{(S\log\Delta S)^k}{k!}\\
        &= e^{S\log\Delta S},
    \end{align*}
    and so $S\log\Delta S = -\log\Delta$. Finally, we have
    \begin{align*}
        S\Delta^{p}S &= S\sum_{k\ge 0}\frac{(p\log\Delta)^k}{k!}S\\
        &= \sum_{k\ge 0}\frac{(-p\log\Delta)^k}{k!}\\
        &= \Delta^{-p}.
    \end{align*}
    $2.$ The first two statements follow from
    \begin{align*}
        J^\dagger &= \Delta^{-1/2}S^\dagger\\
        &= \Delta^{-1/2}S^\dagger S S\\
        &= \Delta^{1/2}S \\
        &= S\Delta^{-1/2}\\
        &= J,
    \end{align*}
    where the second-last line follows from part $1$ with $p=1/2$. The third statement follows from
    \begin{align*}
        J^2 = (S\Delta^{-1/2})(\Delta^{1/2}S) = 1.
    \end{align*}
    $3.$ The first statement follows from $J\log\Delta J = S\Delta^{-1/2}\log\Delta\Delta^{1/2}S = S\log\Delta S = -\log\Delta$ and the second follows from a similar argument.
\end{proof}





\end{document}